\documentclass{llncs}
\usepackage[utf8]{inputenc}
\usepackage{amsmath}
\usepackage{amsfonts}
\usepackage{dsfont}
\usepackage{amssymb}
\usepackage{graphicx}
\usepackage{tikz, pgfplots}
\usetikzlibrary{positioning}
\usetikzlibrary{3d}
\usepackage{subfig}

\def\R{{\mathbf R}}

\def\N{{\mathbf N}}

\def\car#1{{\mathbf 1}}

\def\indic{{\mathbf 1}}
\newcommand{\esp}[1]{{\mathbf E}\left[#1\right]}

\def\P{\mathbf P}

\def\/{\,|\,} 

\def\<{\left\langle}
\def\>{\right\rangle}

\title{Performance analysis of  RIS-assisted communications}
\author{H. Adrat\inst{1} \and L. Decreusefond\inst{2}\orcidID{0000-0002-8964-0957} \and P. Martins\inst{2}\orcidID{0000-0001-9497-4890}}
\institute{University Mohammed VI Polytechnic, Ben Guerir, Morocco \email{hamza.adrat@um6p.ma}\and Telecom
  Paris, Paris, France}

\date{July 2023}
\pgfplotsset{compat=1.18}
\begin{document}
\maketitle{}
\begin{abstract}
	Reconfigurable Intelligent Surfaces (RIS) are currently considered for adoption in future 6G stantards. ETSI and 3GPP have started feasibility and performance investigations of such a technology. This work proposes an analytical model to analyze RIS performance. It relies on a simple street model where obstacles and mobile units are all aligned. RIS is positioned onto a building parallel to the road. The coverage probability in presence of obstacles and concurrent communications is then computed as a performance criteria.
\end{abstract}

\section{Introduction}
In the recent years, there has  been a tremendous amount of  activity in communication
technologies (new waveforms, MIMO signalling, non-orthogonal multiple access and
so on \cite{Basar_2021}) which lead to much improved data rate in wireless 5G/6G systems.
Among the most promising technologies are the reconfigurable intelligent
surfaces (RIS for short). Following \cite{Pan2021} and references therein, an
RIS is a planar surface consisting of an array of passive or active reflecting elements,
each of which can independently  change the phase of the received signal and retransmit
it in  an arbitrary  chosen direction. In other words, radio signals can be tailored to bypass
obstacles between the line of sight between the emitter and the receiver as in Figure ~\ref{fig:1}.

\begin{figure}[!ht]
	\centering
	\includegraphics*[width=0.4\textwidth]{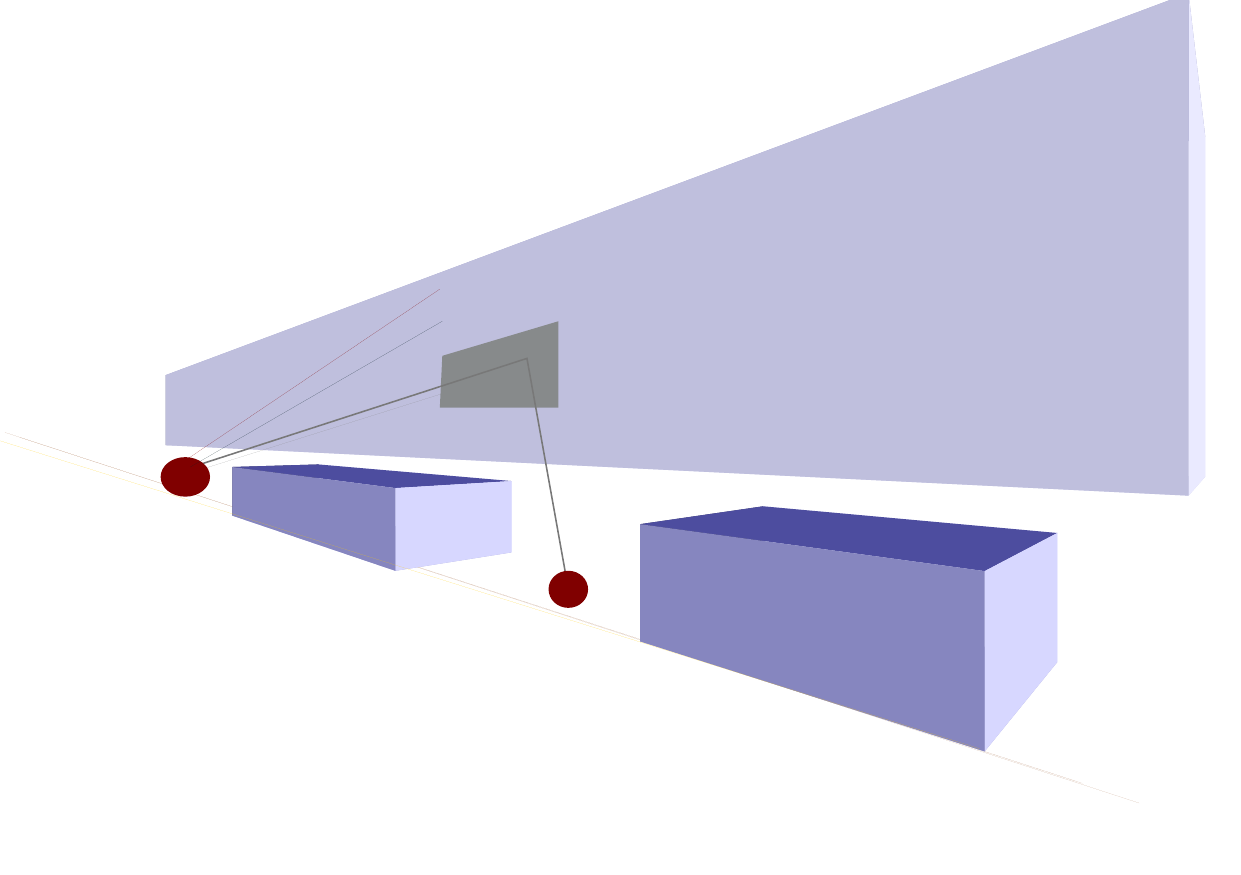}
	\caption{In a RIS-assisted system, radio signals can bypass obstructions.}
	\label{fig:1}
\end{figure}

RIS technology is expected to be applied to wireless systems operating in
frequency bands where the wavelength is of the order of the millimetre. In such
a situation, many objects of daily life are obstacles to the propagation of the
radio waves. On a modeling perspective, this means that we have to consider
models on a meter rather than a hundred meter scale, taking into accounts
buildings, trees, human, cars, etc.

Most of the papers investigating the performance of RIS asssited networks like
\cite{Yang2020,Do2021} are focused on the channel model. At a more macroscopic
level, in \cite{Basar_2021}, for passive RIS,  it is shown that the received power is
proportional to
\begin{equation}
	\label{eq_main:1}
	N^{2}\frac{P_{T}}{d_{\text{SR}}^{2}d_{\text{RD}}^{2}\sigma^{2}}
\end{equation}
where $P_{T}$ is the transmit power, $N$ the number of elements in the RIS,
$\sigma^{2}$ is the receiver noise power and
$d_{\text{SR}}$ and $d_{\text{RD}}$ are the distance between the
source and the RIS, respectively the RIS and the destination, assuming no line-of-sight
between source and destination. The distance between the source and the RIS
significantly affects  the quality of communication, hence the need to find
optimal locations for these devices. Alternatively, we can consider what is the
performance of a system given the location of an RIS which is determined by
practical constraints as trivial as  the possibility of a support or more subtle
legal obligations.

The usual models that come out of stochastic geometry are, in a sense, on a macroscopic scale: In urban areas we work on the scale of a district, in rural areas we work on regions of a few tens of kilometres. The possibility of propagation interruption, multi-paths, etc. are taken care via shadowing and a path-loss exponent larger than $2$. As RIS are supposed to circumvent obstacles, we need to have much finer models at the scale of buildings, cars, or any other  object that can be an obstacle to  wave propagation. This creates a new difficulty in constructing tractable models that include both the position of the RIS and of the obstructions. There are a few papers on  modelling of obstacles. In \cite{Bai2014},
obstacles are represented by rectangles with random centre, length and width. In
\cite{Baccelli2022}, the obstacles are represented as a fractal multiplicative
cascade. In both papers, the
goal is to evaluate the blocking probability in a wireless system due to these
obstacles. The paper that comes  closest to
our consideration is \cite{Sun2023} where a system with multiple base stations
dispatched according to a homogeneous Poisson process is assisted by multiple RISs
is deployed as a  Matern hard core process to take  account of the fact that a RIS
cannot be too close or too far from its serving antenna.  There is no specific
hypothesis regarding the location of obstacles, as the resulting  configuration is assumed
to enable mobile units to avoid any obstacle.

However, none of these papers do model both obstacles and potential support for
a RIS.
This paper addresses this problem in a highly constrained environment.
We consider a road bordered by a building with one RIS on it, serving mobile
units aligned on a line parallel to the wall.
We compute the mean (with respect to the randomness of the environment) number of customers who can be served by a reference user
and the probability that a customer can be served given its position.

The paper is organised as follows. In Section, we compute the mean number of
customers who  can communicate with a typical customer located at the origin
thanks to the RIS. In Section , we take into account the attenuation of the
signal as given in ~\eqref{eq_main:2} to compute the coverage probability at a
given position on the pavement.

\section{Model description}
\label{sec:model}

A natural model should be  three dimensional, but for the sake of simplicity, without loosing
too much information, we restrict our considerations to a planar description.
We work on the infinite line. Obstacles are represented as rectangles of fixed width $d$ and random length.
Between them, there is a portion of free space in which the users may be
located. The lengths of the free space intervals are also random. We denote
by $X(t)$ the random variable which is equal to $1$ if there the point $t$ is
covered by an obstacle and equal to $0$ otherwise. We assume that the process
$X$ is a stationary alternating renewal process and that there is a user at the
origin, i.e. we work given the fact that $X(0)=0$.

Due to the symmetry, we only study the propagation of the signal on the right of
the typical user.

\begin{definition}
	We denote by $(U_{n},\, n\in \N)$, respectively $(W_{n},\, n\in \N)$,
	the successive lengths of time that the system is in state~$0$, respectively
	in state~$1$. According to the assumptions, these random variables are
	independent. The random variables $(U_{n},\, n\ge 2)$ (respectively
	$(W_{n},\, n\in \N)$) are identically
	distributed of cumulative distribution function (cdf) $F_{U}$ and average
	$\gamma_{U}$ (respectively cdf $F_{V}$ and mean $\gamma_{V}$). To ensure
	stationarity, $U_{1}$ is supposed to have probability density function (pdf)  $\gamma_{U}^{-1}(1-F_{U})$.
	We set
	\begin{equation*}
		V_{n}=U_{n}+W_{n},
	\end{equation*}
	the length of the n-th cycle.
\end{definition}
In order for  a customer to be covered by the RIS, it is necessary that at least a
length $\delta$ of the RIS is visible to the user. It is clear that the
rightmost domain which is accessible thanks to the RIS coincides with the
rightmost part of the reconfigurable intelligent surface. We assume that this part is the interval
$[a,a+\delta]$.
The users are assumed to be aligned, at a distance $l$ from the wall on which
lies the RIS.  We denote by
$B_{n}$ (respectively $E_{n}$), the beginning (respectively the end) of the n-th
obstacle. We have
\begin{equation*}
	B_{n}= \sum_{i=1}^{n-1}V_{i}+U_{n} \text{ and }E_{n}=\sum_{i=1}^n V_{i}.
\end{equation*}
The figure~\ref{fig:notations} displays the notations.

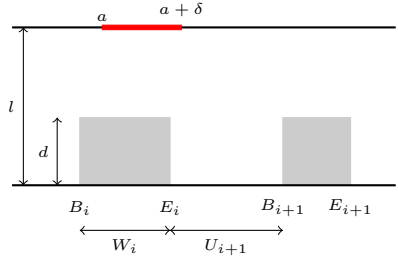
\begin{figure}[!ht]
	\centering
	\begin{tikzpicture}[scale=0.3,font=\fontsize{6}{6}\selectfont]
		\filldraw[color=black!20] (2,0) rectangle (6,3);
		\filldraw[color=black!20] (11,0) rectangle (14,3);
		\draw[thick] (-1,0) -- (16,0);
		\draw[thick] (-1,7) -- (16,7);
		\draw[<->] (-0.5,0) --  node[left] {$l$} (-0.5,7);
		\draw[<->] (1,0) -- node[left] {$d$} (1,3);

		\draw node at (2, -1) {$B_{i}$};
		\draw node at (6, -1) {$E_{i}$};
		\draw node at (11, -1) {$B_{i+1}$};
		\draw node at (14, -1) {$E_{i+1}$};

		\filldraw[red] (3,6.9) node [above,color=black] {$a$} rectangle
		(6.5,7.1) node [above,color=black] {$a+\delta$};
		\draw[<->,thin] (2,-2) -- node[below] {$W_{i}$} (6,-2);
		\draw[<->,thin] (6,-2) -- node[below] {$U_{i+1}$} (11,-2);

	\end{tikzpicture}
	\caption{Notations}
	\label{fig:notations}
\end{figure}

\section{Covered domain}
\label{sec:coverage-domain}
If we assume that the mobile units (or customers) are deployed according to an
homogeneous Poisson process of intensity $\mu$ in the void intervals, the number of customers who
can communicate with the typical user follows a Poisson distribution whose
parameter is $\mu$ times the length of the covered domain. As the positions of
voids and obstacles are random, we compute the mean length of the covered domain with
respect to the law of $X$.
\begin{theorem}
	\label{th:1}
	The mean length of the covered domain $\esp{L}$ is given by:	
		\begin{multline*}
		\esp{L}  = \sum_{i \ge 1} \left[ \int_{a}^{+\infty} \esp{ \left( U_1 - \frac{t-a}{\rho - 1} \right) \indic_{ \left \{ U_1 > \frac{t-a}{\rho - 1} \right \} } } \, f_{i}(t) \, dt \right.                                              \\
		+ \int_0^{a} \esp{ U_1 \indic_{ \left \{ U_1 > a + \delta - t \right \} } } \, f_{i}(t) \, dt                                                                                                                         \\
		+ \left. \int_{0}^a \esp{\left( \frac{\rho}{\rho - 1} U_1 - \frac{a+\delta-t}{\rho - 1} \right) \, \indic_{ \left \{ \frac{a + \delta - t}{\rho} \le U_{1} \le a + \delta - t \right \} }} \, f_{i}(t) \, dt \right],
		\end{multline*}
	where $f_i$ is the probability density function of the random variable $E_i = \displaystyle \sum_{n=1}^i V_n$.
\end{theorem}
To prove this theorem, we must discuss according to the position of the two ends of the RIS with respect to the sequence of obstacles. We have three situations. The most frequent case, illustrated in Figure~\ref{fig:scenario_1}, is the situation  where the leftmost part of the RIS (located at abscissa $a$) is on the left of the end of an obstacle. The next case occurs only once and is obtained when $[a,a+\delta]$ lies in between the end of an obstacle and the beginning of the next. The complementary scenario which appears a finite number of times is described in Figure~\ref{fig:scenario_3}.

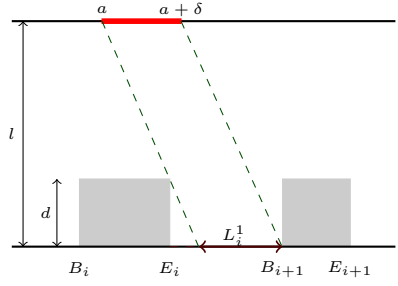
\begin{figure}[!ht]
	\centering
	\begin{tikzpicture}[scale=0.3,font=\fontsize{6}{6}\selectfont]
		\filldraw[black!20] (2,0) rectangle (6,3);
		\filldraw[black!20] (11,0) rectangle (14,3);

		\draw[thick] (-1,0) -- (16,0);
		\draw[thick] (-1,10) -- (16,10);
		\draw[<->] (-0.5,0) -- (-0.5,10);
		\draw node at (-1, 5) {$l$};
		\draw[<->] (1,0) -- (1,3);
		\draw node at (0.5, 1.5) {$d$};

		\draw node at (2, -1) {$B_{i}$};
		\draw node at (6, -1) {$E_{i}$};
		\draw node at (11, -1) {$B_{i+1}$};
		\draw node at (14, -1) {$E_{i+1}$};

		\filldraw[red] (3,9.9) rectangle (6.5,10.1);
		\draw node at (3, 10.5) {$a$};
		\draw node at (6.5, 10.5) {$a+\delta$};

		\draw[dashed, green!30!black] (6.5,10) -- (11,-0.05);
		\draw[dashed, green!30!black] (3,10) -- (7.3,-0.05);

		\draw[dashed, red!30!black] (6,0) -- (7.3,0);
		\draw[<->, thick, red!30!black] (7.3, 0) -- (11, 0);
		\draw node at (8.85, 0.5) {$L^1_i$};
	\end{tikzpicture}
	\caption{Covered domain between $i^{th}$ and $(i+1)^{th}$ obstacles - Scenario 1}
	\label{fig:scenario_1}
\end{figure}

\begin{figure}[!ht]
	\centering
	\begin{tikzpicture}[scale=0.3,font=\fontsize{6}{6}\selectfont]
		\filldraw[black!20] (2,0) rectangle (6,3);
		\filldraw[black!20] (11,0) rectangle (14,3);

		\draw[thick] (-1,0) -- (16,0);
		\draw[thick] (-1,10) -- (16,10);
		\draw[<->] (-0.5,0) -- (-0.5,10);
		\draw node at (-1, 5) {$l$};
		\draw[<->] (1,0) -- (1,3);
		\draw node at (0.5, 1.5) {$d$};

		\draw node at (2, -1) {$B_{i}$};
		\draw node at (6, -1) {$E_{i}$};
		\draw node at (11, -1) {$B_{i+1}$};
		\draw node at (14, -1) {$E_{i+1}$};

		\filldraw[red] (9.5,9.9) rectangle (13,10.1);
		\draw node at (9.5, 10.5) {$a$};
		\draw node at (13, 10.5) {$a+\delta$};

		\draw[dashed, green!30!black] (9.5,10) -- (6,-0.05);
		\draw[dashed, green!30!black] (13,10) -- (10.2,-0.05);

		\draw[<->, thick, red!30!black] (6, 0) -- (10.2, 0);
		\draw node at (8.1, 0.5) {$L^3_i$};
	\end{tikzpicture}
	\caption{Covered domain between $i^{th}$ and $(i+1)^{th}$ obstacles - Scenario 3}
	\label{fig:scenario_3}
\end{figure}
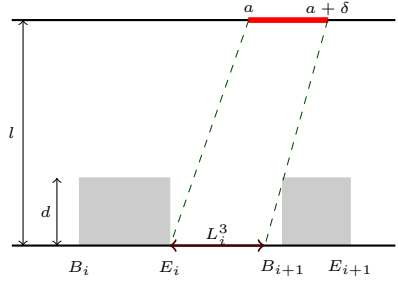
\begin{lemma}[Scenario 1]
	\label{lem:scenario_1}
	For $i \ge 1$, if $a \le E_i$, then the mean length of the covered domain $\esp{L^1_i}$ between $i^{th}$ and $(i+1)^{th}$ obstacles is :
	\begin{equation*}
		\esp{L^1_i} = \int_{a}^{+\infty} \esp{ \left( U_1 - \frac{t-a}{\rho - 1} \right) \indic_{ \left \{ U_1 > \frac{t-a}{\rho - 1} \right \} } } \, f_{i}(t) \, dt.
	\end{equation*}
\end{lemma}
Note that since $E_k>E_i$ for $k>i$, the condition $a\le E_k$ is satisfied eternally from index $i$ onwards and thus Scenario~1 is the most common.
\begin{proof}[Proof of Lemma~\protect\ref{lem:scenario_1}]
	According to Figure~\ref{fig:scenario_1} and elementary geometry, we derive  the expression of $L^1_i$:
	\begin{equation*}
		\begin{split}
			L^1_i & = \left( \left( B_{i+1} - E_i \right) - \dfrac{d}{l-d} \left( E_i - a \right) \right)^+  \indic_{ \{ E_i \ge a \} } \\
			      & = \left( (B_{i+1} - E_i) - \dfrac{1}{\rho - 1}(E_i - a) \right)^+ \indic_{ \{ E_i \ge a \} }                        \\
			      & = \left( U_{i+1} - \dfrac{1}{\rho - 1}(E_i - a) \right) \indic_{ \{ a \le E_i \le a + (\rho - 1) U_{i+1} \} }.
		\end{split}
	\end{equation*}
	Thus,  the mean length of the covered domain between $i^{th}$ and $(i+1)^{th}$ obstacles in this scenario is given by:
	\begin{align*}
\esp{L^1_i}
			 & = \esp{\left( U_{i+1} - \dfrac{1}{\rho - 1}(E_i - a) \right) \indic_{ \{ a \le E_i \le a + (\rho - 1) U_{i+1} \} }}                                                 \\
			            & = \int_{\R} \esp{ \left( U_1 - \frac{t-a}{\rho - 1} \right) \indic_{ \left \{ U_1 > \frac{t-a}{\rho - 1} \right \} } } \, f_{i}(t) \indic_{ \{ t > a \} } \, dt \\
			            & = \int_{a}^{+ \infty} \esp{ \left( U_1 - \frac{t-a}{\rho - 1} \right) \indic_{ \left \{ U_1 > \frac{t-a}{\rho - 1} \right \} } } \, f_{i}(t) \, dt.
\end{align*}
	The proof is thus complete.
\end{proof}
The proofs for Scenario 2 and 3 are similar. We obtain the following lemma.
\begin{lemma}[Scenario 2]
	\label{lem:scenario_2}
	For $i \ge 1$, if $a \ge E_i$ and $a+\delta \le B_{i+1}$, then the mean length of the covered domain $\esp{L^2_i}$ between $i^{th}$ and $(i+1)^{th}$ obstacles is :
	\begin{equation*}
		\esp{L^2_i} = \int_0^{a} \esp{ U_1 \indic_{ \left \{ U_1 > a + \delta - t \right \} } } \, f_{i}(t) \, dt
	\end{equation*}
\end{lemma}
\begin{lemma}[Scenario 3]
	\label{lem:scenario_3}
	For $i \ge 1$, if $a \ge E_i$ and $a+\delta \ge B_{i+1}$, then the mean length of the covered domain $\esp{L^3_i}$ between $i^{th}$ and $(i+1)^{th}$ obstacles is :
		\begin{multline*}
			\esp{L^3_i} =\\\int_{0}^a \mathbf E\left[\left( \frac{\rho U_1-(a+\delta-t)}{\rho - 1}  \right)
			\, \indic_{ \left \{ \frac{a + \delta - t}{\rho} \le U_{1} \le a + \delta - t \right \} }\right] \, f_{i}(t) \, dt.
		\end{multline*}

\end{lemma}
To gain more insights of the previous formula, we instantiate it for the specific case where holes and obstacles are exponentially distributed: We assume that $U_1$ follows an exponential distribution with parameter $\gamma_1$, and $W_1$ follows an exponential distribution with parameter $\gamma_2$.
In this case, the random
variable $$E_i = \displaystyle \sum_{n=1}^i V_n =  \displaystyle \sum_{n=1}^i U_n + \displaystyle \sum_{n=1}^i W_n$$
is expressed as the sum of two independent gamma-distributed random variables
with parameters $(i, \gamma_1)$ and $(i, \gamma_2)$ respectively. We then  need to
introduce the notion of Kummer's confluent hypergeometric function.
\begin{definition}
	Let $M(a,b,z)$ be the Kummer's (confluent hypergeometric) function.
	If $\Re(b) > \Re(a) > 0, \; M(a,b,z)$ can be represented as an integral:
	\begin{equation}\label{eq:kummerfunction}
		M(a,b,z) = \frac{\Gamma(b)}{\Gamma(a) \Gamma(b-a)} \int_0^1 e^{zt} \, t^{a-1} \, (1-t)^{b-a-1} \, dt \, ,
	\end{equation}
	where $\Gamma$ is the usual Gamma function.
\end{definition}
With these notations at hand, we have:
\begin{equation*}
	f_{i}(t) = \frac{(\gamma_1 \gamma_2)^i}{(2i-1)!} \, t^{2i-1} \, e^{- \gamma_2 t} \, M(i, 2i, (\gamma_2 - \gamma_1)t).
\end{equation*}
\begin{theorem}
	\label{th:2}
	Using the previous notations, the mean length of the covered domain $\esp{L}$ in this case is given by:
	\begin{align*}
\esp{L}
			&=\frac{1}{\gamma_1} e^{\frac{\gamma_1 a}{\rho - 1}} \frac{1}{1-r} \\
& - \frac{1}{\gamma_1} \sum_{i \ge 1} \frac{(\gamma_1 \gamma_2)^i}{(2i-1)!} \left[ \int_0^{a} e^{-\gamma_1\frac{t-a}{\rho-1}} \, t^{2i-1} \, e^{-\gamma_2 t} \, M(i, 2i, (\gamma_2 - \gamma_1)t) \, dt \right.\\
&+ \frac{e^{-\gamma_1(a+\delta)}}{\rho-1} \int_0^a e^{-(\gamma_2 - \gamma_1)t} \, t^{2i-1} \, M(i, 2i, (\gamma_2 - \gamma_1)t) \, dt\\
&- \left. \frac{\rho \, e^{-\frac{\gamma_1}{\rho}(a+\delta)}}{\rho-1} \int_0^a e^{-(\gamma_2 - \frac{\gamma_1}{\rho})t} \, t^{2i-1} \, M(i, 2i, (\gamma_2 - \gamma_1)t) \, dt \right],
\end{align*}
	where $r = \dfrac{1}{\left( 1 + \frac{1}{\rho - 1} \right) \left( 1 + \frac{\gamma_1}{\gamma_2(\rho - 1)} \right)} < 1$.
\end{theorem}
Under  the condition $\gamma_1 a \ll \rho - 1$ which means that scenarios $2$ and $3$ are hardly achievable
(especially when the value of $a$ is small), we can neglect all terms associated with $\int_0^a f_i(t) \, dt$.
\begin{corollary}
	\label{cor:approximation}
	If $\gamma_1 a \ll \rho - 1$, then
	\begin{equation*}
		\esp{L} \simeq \frac{1}{\gamma_1}  \frac{1}{1-r}\cdotp
	\end{equation*}
\end{corollary}

\section{Coverage probability}
\label{sec:coverage-probability}
We now consider that there are many UE on the line which want to communicate
with the user located at the origin. There is a special customer at position $x$
and we want to evaluate the probability that she can communicate with the
origin, via the RIS, considering all the other communications as interference.
For the sake of simplicity, we now consider the RIS as a point located at
position $a$: We no longer take into account the necessity that a length
$\delta$ of the RIS is visible by the UE. We borrow the following results from \cite{Basar_2021}. When the RIS is of the active sort (which requires power supply), the received
power is proportional to
\begin{equation}
	\label{eq_main:2}
	N \frac{P_{T}P_{A}}{P_{A}\sigma_{v}^{2}d_{\text{SR}}^{2}+ P_{T} \sigma^{2}d_{\text{RD}}^{2}},
\end{equation}
where $P_{A}$ is the maximum RIS-reflect power and  $\sigma_{v}^{2}$ is the RIS-induced noise power. As usual, we  incorporate in the previous formula  the Rayleigh fading, represented by $F_x$, an independent exponentially distributed random variable. Additionally, we have $d_{\text{SR}}^{2} = l^2 + a^2$ and $d_{\text{RD}}^{2} = l^2 + (x-a)^2$.
Consequently, we retrieve the formula for the received power from a given transmitter $x$ to $0$ as follows:
\begin{equation}
	\label{eq_power_received}
	P_r(O \leftarrow x) = \frac{c F_x}{K + (x-a)^2},
\end{equation}
where $c = \dfrac{N P_A}{\sigma^2}$ and $K = \dfrac{P_A \sigma_v^2}{P_T \sigma^2}(l^2 + a^2) + l^2$. \\
It should be noted that in ~\eqref{eq_power_received}, it is assumed that the position of transmitter $x$ allows him to communicate with $O$, 
i.e., $x$ is positioned between two obstacles and is well covered by the active RIS. \\
We denote by $\Phi$ the Poisson process with intensity $\lambda$ representing
the set of all transmitters $y$ that may be interfering with the communication
between $x$ and $O$. Recall that  $X(y)=0$ means that  $y$ is located between two
obstacles. We denote  by $\Phi_{X}$, the points $y$ of $\Phi$
for which $X(y)=0$. 
The process $X$ and and the point process $\Phi$ are assumed to be independent but the random variables
$X(y)$ are not independent so strictly speaking, we cannot say that  $\Phi_{X}$ is a Poisson
point process. 
Furthermore, according  to~\eqref{eq_power_received}, for $y\in \Phi_{X}$, we have:
\begin{equation*}
	P_r(O \leftarrow y) = \frac{c F_y}{K + (y-a)^2} \, \indic_{\{\tau_y \ge \frac{y-a}{\rho}\}},
\end{equation*}
where $\tau_y$ is the distance between $y$ and the last obstacle before $y$ 
and the condition $\tau_y \ge \frac{y-a}{\rho}$ amounts to saying that $y$ is
covered by the active RIS. This last condition creates another theoretical  difficulty s: Since the renewal process we use as a model is stationary, it is
clear that $\tau_y$ is exponentially distributed with parameter~$\gamma_1$ but
for different $y$ and $z$, the random variables $\tau_y$ and $\tau_z$ are not
independent. However,  for the sake of tractability, we consider that $\Phi_{X}$ is a
Poisson process  of intensity
$\lambda_{\Phi_{X}}=\lambda \gamma_{1}(\gamma_{1}+\gamma_{2})^{-1}$ and  that the random variables $(\tau_y)_{y \in \Phi}$ are independent (we
call these assumptions $H_{0}$). We show below by simulation that these assumptions turns out to be harmless.
With these notations at hand, we denote the SINR at $x$ by : 
\begin{equation*}
	\mathrm{SINR}_x = \frac{P_r(O \leftarrow x)}{\displaystyle \sum_{y \in \Phi_{X}} P_r(O \leftarrow y)},
\end{equation*}
and then we have the following theorem:
\begin{theorem}
	\label{th:sinr}
	For a threshold $\theta > 0$, under $H_{0}$, we have :
	\begin{equation*}
\P \left( \mathrm{SINR}_x \ge \theta \,|\, a \leftarrow x\right) = \exp \left( - \frac{\beta \lambda}{\sqrt{K + \beta}} \int_{0}^{+ \infty} \frac{1}{1 + y^2} e^{ - \frac{\gamma_1}{\rho} \sqrt{K + \beta} \, y} \, \mathrm{d}y \right),
\end{equation*}
	where $\beta = \theta \left( K + (x-a)^2 \right)$.
\end{theorem}
\begin{proof}[Proof of Theorem~\protect\ref{th:sinr}]
	Using the previous notations, we have:
	\begin{multline*}
			\P \left( \mathrm{SINR}_x \ge \theta\,|\, a \leftarrow x \right) \\
			= \P \left( \frac{c F_x}{K + (x-a)^2} \ge  \displaystyle \sum_{y \in \Phi_{X}} \frac{c\theta F_y}{K + (y-a)^2} \, \indic_{\{\tau_y \ge \frac{y-a}{\rho}\}} \right) \\
			= \esp{ \exp \left( - \beta \displaystyle \sum_{y \in \Phi_{X}} \frac{F_y}{K + (y-a)^2} \, \indic_{\{\tau_y \ge \frac{y-a}{\rho}\}} \right) }
		  .
	\end{multline*}
Under $H_{0}$, we then recognise the probability generating functional of the  Poisson process
$\Phi_{X}$  with a pair of independent marks both following an  exponential law
of parameter~$1$.
	We get:
	\begin{equation*}
\P \left( \mathrm{SINR}_x \ge \theta\,|\, a\leftarrow x \right)
			= \exp \left( - \lambda \int_{0}^{+\infty} \frac{\beta}{K + \beta + y^2} e^{- \frac{\gamma_1}{\rho} y} \mathrm{d}y \right).
\end{equation*}
The final form  follows by a change of variable.
\end{proof}

\section{Numerical analysis}
\label{sec:numerical-analysis}


Using the approximation of Corollary ~\ref{cor:approximation}, we have:
	\begin{equation*}
		\esp{L} \simeq \frac{1}{\gamma_1}  \frac{\rho \left( \alpha + \frac{1}{\rho - 1}\right)}{1 + \alpha + \frac{1}{\rho - 1}}, 
	\end{equation*}
	where $\alpha = \dfrac{\gamma_2}{\gamma_1}$.
Note that $\alpha$ is the ration of the  mean length of holes to the mean length of obstacles: A large value of $\alpha$ means that obstructions are small compared to the quantity of empty space. The intuition then says that the RIS is likely to be very efficient in such a situation. This is what we recover here as $\esp{L}$ strongly increases for the smallest increments of~$\alpha$, see Figure~\ref{fig:asymptotic}. 
\begin{figure}[!ht]
	\centering
	\includegraphics*[width=0.4\textwidth]{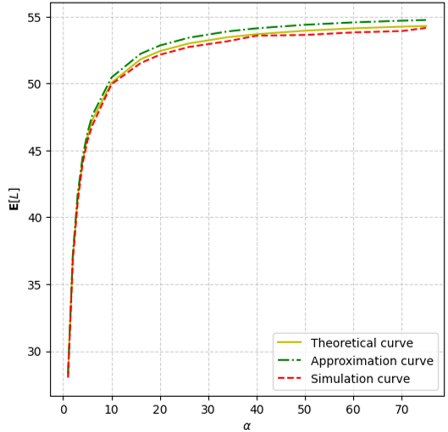}
	\caption{$\esp{L}$ as a function of $\gamma_2/\gamma_1$.}
	\label{fig:asymptotic}
\end{figure}

The main idea in Theorem~\protect\ref{th:sinr} is that we have supposed that
$\Phi_{X}$ is a Poisson process and that
the  $(\tau_y)_{y \in \Phi}$ are independent. In order to verify the
compatibility  of these assumptions with a more realistic model, we compute by
simulation the quantity $\P \left( \mathrm{SINR}_x \ge \theta\,|\, O\leftarrow x \right)$ for values of $\theta$ ranging from $0.1$ to $25$ without the hypothesis of independence. As we have two sources of randomness: obstacles and void spaces on the one hand, locations of the users on the other hand, we must say what varies and what is fixed.
We  fix the position of $x$ as well as the Poisson process of the other
transmitters (i.e. $\Phi$).  At each iteration, we generate the obstacles (i.e. $X$) and
keep  only the configurations for which  $x$ is positioned between two obstacles
 and is  covered by the active RIS. We compute the  average of $(\mathrm{SINR}_x \ge \theta)$ only on these configurations. For numerical application, we consider that $P_T = P_A = 20$ dBm and $\sigma_v^2 = \sigma^2 = -90$ dBm,
and without loss of generality, we take $a = 0$, $\rho = 20$, $\gamma_1 = 0.5$ and $\lambda = 0.2$.
The results  given in Figure ~\ref{fig:sinr_comparaison} show that the coverage
probability without interference is very close to the coverage probability
taking into account dependency between the $\tau_y$'s. Moreover, the curve
obtained under the  hypothesis  $H_0$ is lower that the
true curve, which means that choosing the parameters in order to guarantee a
coverage probability greater than a given threshold under $H_0$ ensures that the
real coverage probability will be higher. Such a beneficial effect of
correlations has already been observed in \cite{lee2018effect}.

\begin{figure}[!ht]
	\centering
	\includegraphics*[width=0.5\textwidth]{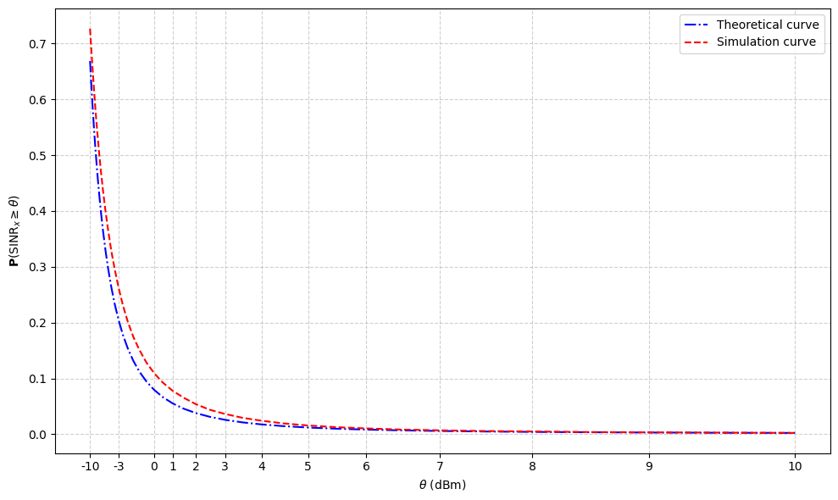}
	\caption{Comparison results of the coverage probability}
	\label{fig:sinr_comparaison}
\end{figure}


\section*{Funding}
\label{sec:funding}
The second and third authors were supported in part by the French National Agency for
Research (ANR) via the project n°ANR-22-PEFT-0010 of the France 2030 program PEPR
réseaux du Futur.

\end{document}